\pdfoutput=1
\documentclass[11pt]{article}
\usepackage[margin=1in]{geometry}
\usepackage{authblk}
\usepackage{amsmath,amsthm,amsfonts}
\usepackage{hyperref,graphicx,subcaption}
\usepackage[noend]{algorithmic}
\usepackage{cite}
\usepackage{url}

\newtheorem{lemma}{Lemma}
\newtheorem{theorem}{Theorem}

\title{Scheduling Autonomous Vehicle Platoons Through 
       an Unregulated Intersection}

\author[1]{Juan Jos\'e~Besa~Vial}
\author[1]{William E.~Devanny}
\author[1]{David Eppstein}
\author[1]{Michael T.~Goodrich}

\affil[1]{Computer Science Department, University of California, Irvine, USA}

\begin{document}

\maketitle

\begin{abstract}
We study various versions of the problem of scheduling
platoons of autonomous vehicles through an unregulated intersection,
where an algorithm must schedule which 
platoons should wait so that others can go through, so as to minimize
the maximum delay for any vehicle.
We provide polynomial-time algorithms for constructing
such schedules for a $k$-way merge intersection,
for constant $k$, and for a crossing intersection involving two-way traffic.
We also show that the more general problem of
scheduling autonomous platoons through an intersection that includes
both a $k$-way merge, for non-constant $k$, and a crossing of two-way traffic
is NP-complete.
\end{abstract}

\section{Introduction}
The advent of autonomous vehicles is introducing a number
of interesting algorithmic questions concerned with how to
coordinate the motion of such vehicles, especially through
unregulated intersections (e.g., 
see~\cite{A6859163,altche2016time,%
au2010motion,Berger:2010,B1411418,C6728285,Dasler2015,Dresner:2004,%
Dresner:2005,dresner2008multiagent,F1429220,%
IlginGuler2014121,L5940562,L6121907,L6957676,M7039600,N683216,%
VanMiddlesworth:2008,Wuthishuwong2015}).
Such an intersection would not have any stop signs or lights and 
would instead rely on algorithmic coordination 
between the autonomous vehicles
approaching the intersection in order to 
prevent collisions.

In addition to intersection management, another interesting algorithmic
development for autonomous vehicle
control is the use of \emph{platoons}, where
a sequence of autonomous vehicles operates in close proximity, 
much like the cars of a locomotive train, so as to save time and/or energy.
(E.g., see~\cite{A4020351,IlginGuler2014121,Rajamani200115,S1468255}.)
Ideally, we would like to keep platoons as contiguous sequences of vehicles,
even as they are traveling through an intersection.

Thus, we are interested in this paper in algorithms 
for solving the problem of scheduling autonomous vehicle platoons through
an unregulated intersection, so as to minimize the maximum delay 
for any vehicle (due to waiting in traffic at the intersection) while 
keeping platoons as contigous sequences of vehicles.
The algorithms we describe are agnostic about whether times and locations are continuous variables or (following the discrete framework of Dasler and Mount~\cite{Dasler2015}) discretized
to be integers, which we may consider as normalized such that each platoon moves one unit of distance in one unit of time.
No two vehicles are allowed to occupy the same point at the same time,
but platoons advance in ``lock step'', with all vehicles within a platoon moving the same distance as each other in each time unit.
For continuous models of time and space, we obtain strongly polynomial time bounds,
and for discrete models we obtain time bounds that are polynomial both in the number of platoons and in the logarithm of the total travel time.

Although online scheduling algorithms would also be of interest,
in this paper, we focus on the offline scheduling problem, where
we are given in advance the location and path for each platoon
wishing to travel through a given intersection.

\subsection{Related Work}
Prior related work on autonomous vehicle coordination through 
an intersection is usually referred to as 
\emph{autonomous intersection management}, with most of the previous
work focused on low-level sensor, multi-agent, and acceleration/braking control
algorithms (e.g., see~\cite{Dresner:2004,Dresner:2005,dresner2008multiagent,%
L6121907,L6957676,M7039600,VanMiddlesworth:2008})
or high-level management policies and strategies
(e.g., see~\cite{A6859163,C6728285,N683216,Wuthishuwong2015}).

Closer to the mid-level approach that we take in this paper,
Guler {\it et al.}~\cite{IlginGuler2014121} study the
problem of scheduling platoons through an unregulated intersection, but they
focus on the problem of minimizing the total number of stops 
for all vehicles or the total
delay for all vehicles, e.g., in simple first-in/first-out strategies,
rather than minimizing the maximum delay for any vehicle.
Also in this mid-level framework,
Dasler and Mount~\cite{Dasler2015} 
build on the work of Berger and Klein~\cite{Berger:2010} for solving a
geometric version of the ``Frogger'' video game to
study the problem of routing
variable-length cars (which could also model platoons) 
through multiple intersections.
They introduce a discrete model for autonomous vehicle scheduling, which,
as we mentioned above, we use in this paper. 
They show that the problem of scheduling vehicles through an arbitrary grid
configuration of multiple intersections to minimize the maximum
delay for any vehicle is NP-complete. 
Such a result is also implied by the work of
Hatzack and Nebel~\cite{hatzack2014operational}
on modeling traffic scheduling as job-shop scheduling with blocking.
These hardness results do not apply to 
the single intersection problem that we study in this paper, however, 
because their proofs require interactions between multiple intersections.
Dasler and Mount~\cite{Dasler2015} 
also give several polynomial-time
algorithms for special cases in which horizontally-traveling vehicles must
always yield to and never block vertically-traveling vehicles, but these
algorithms
similarly do not apply to the problems we study in this paper, since we don't 
assign different priorities to different platoons.

In the problems that we study in this paper, platoons must always move
monotonically, that is, they may move forward or stop, but they may not
back up. If platoons are allowed both forward and backward movements, then
path planning becomes much harder. See, e,g., the PSPACE-completeness
proofs of Hearn and Demaine for 
various traffic-clearing problems~\cite{HEARN200572}.

\subsection{Our Contributions}
In this paper,
we study various versions of the problem of scheduling
platoons of autonomous vehicles through an unregulated intersection,
where the set of such platoons and their paths are given in advance.
The optimization goal in the problems that we study is 
to minimize the maximum delay for any vehicle.
We provide polynomial-time algorithms for constructing
such schedules for a $k$-way merge intersection,
for constant $k$, and for a crossing intersection involving two-way traffic.
Our solutions are based on novel uses of dynamic
programming and parametric search techniques.

We also show that the more general problem of scheduling autonomous platoons 
through an intersection that includes
both a $k$-way merge, for non-constant $k$, and 
a crossing of two-way traffic is NP-complete,
via a reduction from the partition problem,
which is known to be NP-complete 
(e.g., see Garey and Johnson~\cite{Garey:1990:CIG:574848}).


\section{Definitions}
An intersection may be modeled as a collection of incoming and outgoing traffic lanes, together with constraints on which pairs of incoming and outgoing lanes can be used for simultaneous traffic flows without interference with each other. Each platoon can be specified by the incoming and outgoing lanes it follows, together with the times that the start and end of the platoon would reach the intersection if no delays are imposed; by analogy to job shop scheduling, we call the time at which the start of the platoon would reach the intersection the \emph{release time} of the platoon. The \emph{length} of a platoon is the difference between its start and end times. We require that the platoons initially occupy disjoint positions on each of their incoming lanes (that is, on each lane, the start and end times of each platoon form disjoint ranges of time) and that they remain disjoint throughout any valid schedule of traffic: two platoons on the same lane cannot exchange positions. A platoon cannot be subdivided into smaller units of traffic.

A schedule for such a problem can be described by specifying the time that each platoon begins crossing the intersection, which we call the \emph{crossing time} of the platoon. The time that it finishes crossing is the crossing time plus the length.
A schedule is \emph{valid} if it meets the following conditions:
\begin{itemize}
\item Every platoon's crossing time is on or after its release time. (Platoons can't break the speed limit to reach the intersection more quickly.)
\item For every two platoons on the same incoming lane, the crossing time of the second platoon is on or after the release time plus length of the first platoon. (Platoons in the same lane can't pass each other.)
\item If two platoons are on incompatible pairs of incoming and outgoing lanes, the open intervals between their crossing times and crossing times plus lengths are disjoint. (Cross traffic should not collide.)
\end{itemize}

The \emph{delay} of any platoon, in a valid schedule, is the difference between its crossing time and its release time (equivalently, the difference between the time that the end of the platoon finishes crossing in the schedule and the time that the end would finish crossing if there were no delays). The delay of a valid schedule is the maximum of the delays of the platoons. Our goal is to find a valid schedule with minimum delay.
(See Figure~\ref{fig:defs}.)

\begin{figure}[hb!]
\centering
\begin{tabular}{ccc}
        \includegraphics[width=.3\textwidth]{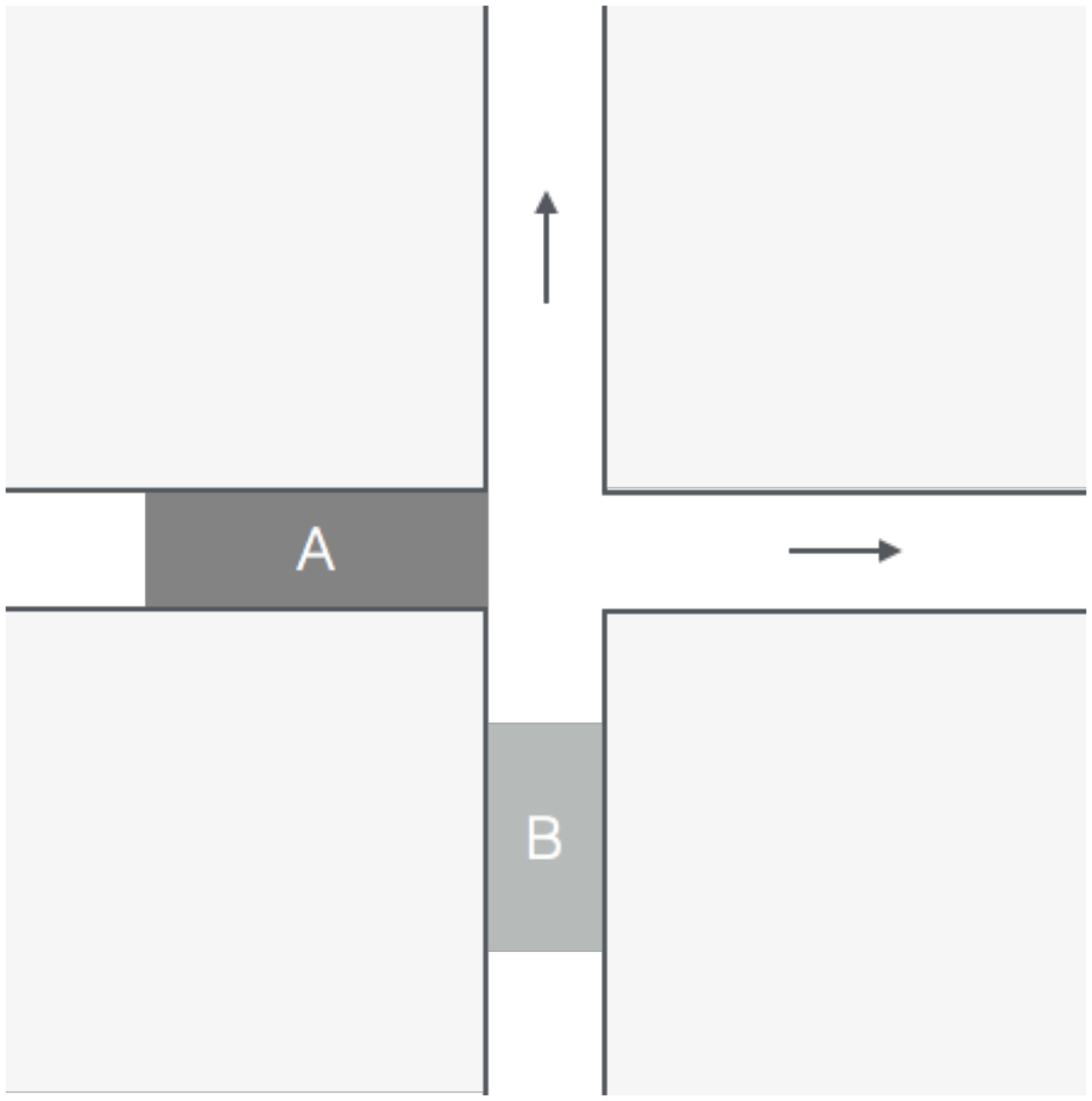} 
&
        \includegraphics[width=.3\textwidth]{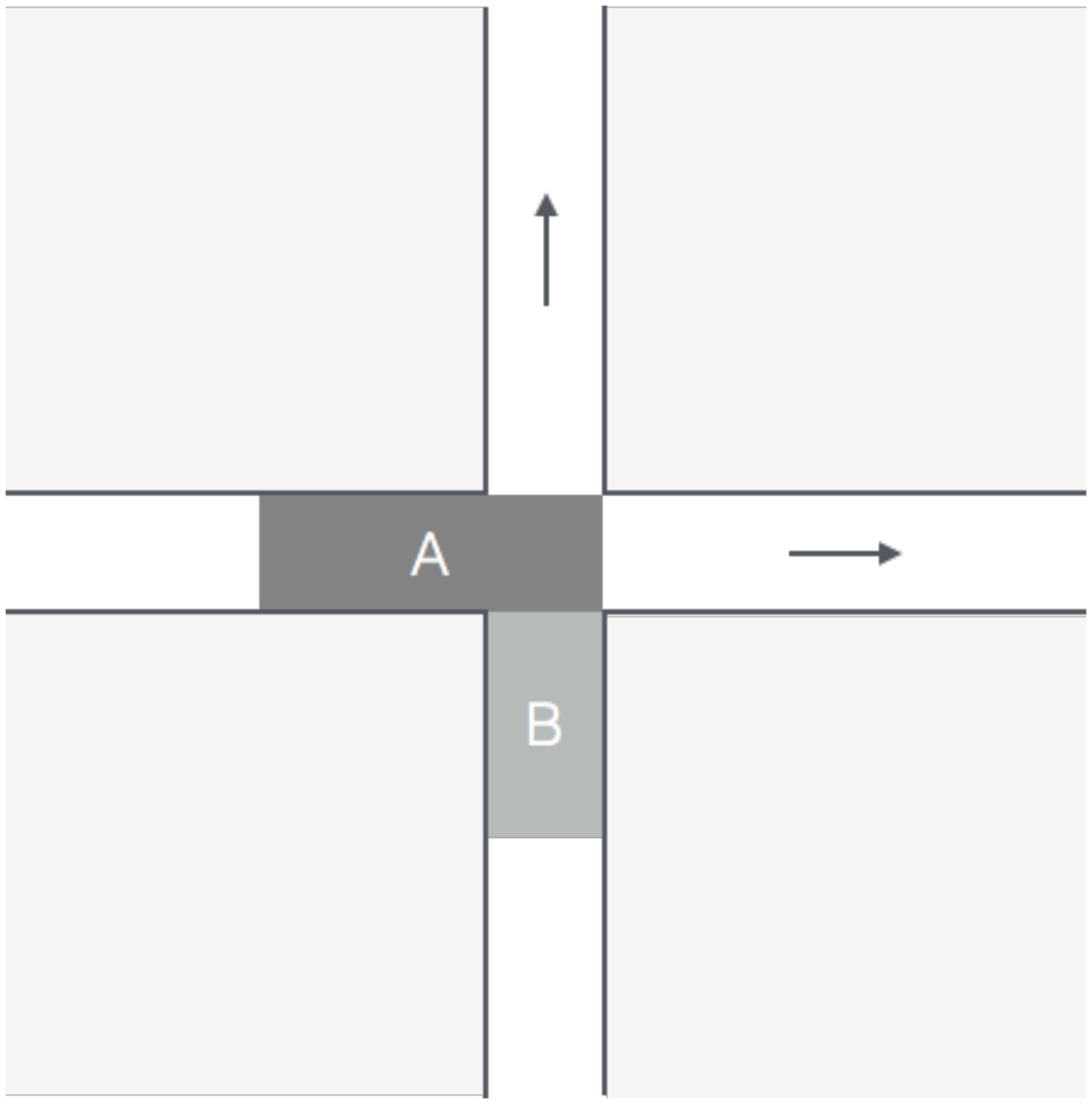}
&
        \includegraphics[width=.3\textwidth]{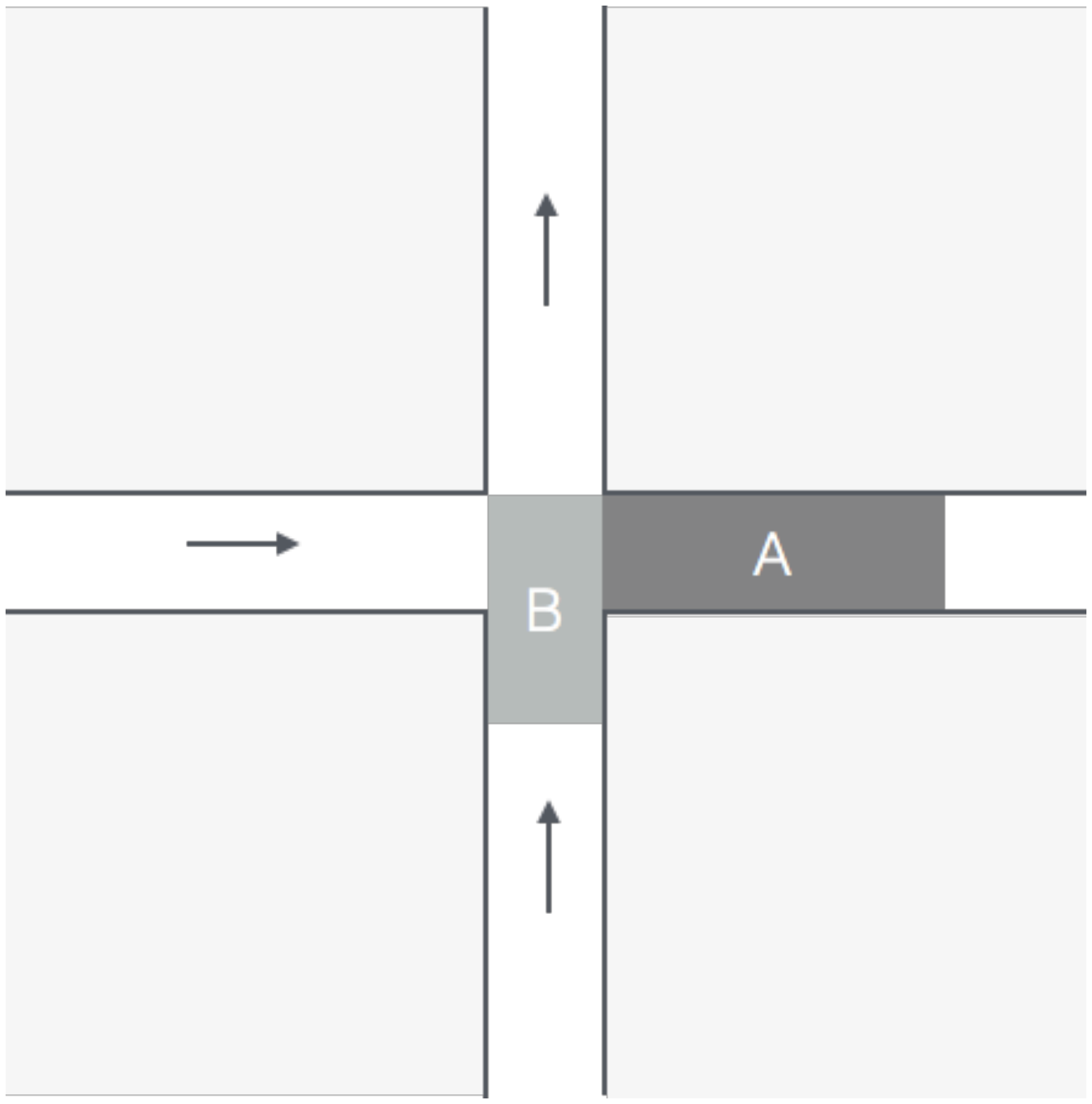}
\\[-0.4in]
$t=0$
&
$t=1$
&
$t=3$
\end{tabular}
\caption{At $t=0$ platoon $A$, of length 3, reaches the intersection and begins to cross it. The release time of $A$ is 0 and it's delay is also 0. Later at $t=1$ platoon $B$ arrives at the intersection; its release time is 1 but it cannot cross immediately because $A$ is in the intersection. Finally at $t=3$ $B$ begins to cross, with delay $2$. The delay of the overall schedule of these platoons is the maximum of the delays of the two platoons, 2. }
\label{fig:defs}
\end{figure}

The main parameter in our analysis will be the \emph{size} of a scheduling problem, which we define to be the number of platoons and which we will usually denote by the variable~$n$.
In some cases the analysis of our algorithms will also depend on the numerical resolution of the input. If all release times are integers, then there necessarily exists an optimal schedule in which the crossing times are also integers. In this case we define the \emph{length} of a schedule to be the maximum release time plus the sum of the lengths of all platoons. This number, which we denote by~$L$, provides a na\"ive bound on the maximum time required by a valid schedule that does not introduce gratuitous delays.


\section{Polynomial-time Algorithms}

In this section, we provide algorithms for platoon scheduling,
whose time bounds are either \emph{strongly polynomial} (i.e., with a runtime that depends polynomially on the number of platoons, but not at all on the timing of the platoons)
or \emph{polynomial} (depending polynomially on the number of platoons and on the number of bits of precision needed to specify their timing).
In contrast, algorithms whose running time includes terms proportional to the total number of time units of the schedule are not polynomial.

\subsection{Parametric search}
We will provide different algorithms for different intersection models, but they will all be based on an algorithmic metaprinciple, the \emph{parametric search} technique of Megiddo~\cite{Meg-JACM-83}, which allows us to convert decision algorithms (which answer whether or not there is a schedule with a given delay) into optimization algorithms (which find a schedule with minimum delay).

More precisely, we define a decision algorithm for a platoon scheduling problem to be an algorithm that takes as input a scheduling task and a parameter $d$, and tests whether there exists a valid schedule whose delay is at most $d$. We require that the only use the algorithm makes of its parameter $d$ is to perform a comparison, $d\ge c$, of $d$ against another number, $c$, calculated from the other input values (not including $d$). Intuitively, the decision algorithm is allowed to test questions like ``if this platoon were to cross now, would it cause other platoons' delays to exceed the given delay parameter?'' We will use $D(n)$ to denote the running time for such an algorithm.
However, as well as performing this algorithm directly, with a numeric parameter $d$, we will also
simulate the algorithm on a parameter that is not given to it explicitly, by performing some alternative computation to replace the comparisons with $d$. As a simple example, we have the following simulation.

\begin{lemma}
\label{lem:modified-decision}
If there is a decision algorithm (as described above) that either finds a valid schedule with delay at most $d$ or determines that no such schedule exists, in time $D(n)$, then there is also an algorithm that tests for a given $d$ whether there is a valid schedule with delay strictly less than $d$, in time $O(D(n))$.
\end{lemma}

\begin{proof}
We simulate the decision algorithm on the parameter $d-\epsilon$, for an unknown number $\epsilon>0$ that is smaller than the difference in delays between $d$ and the best valid schedule.
To do so, every time the simulated algorithm needs to test whether $d\ge c$, for some computed number $c$, we substitute the result of the comparison $d>c$.
\end{proof}

Using this method of simulation, we can transform a decision algorithm into an optimization algorithm, as follows.

\begin{lemma}
\label{lem:parametric-search}
Let $n$ denote the size and $L$ denote the length of a platoon scheduling problem.
Suppose that there exists a decision algorithm (as described above) that takes time $D(n)$ to test whether there is a schedule whose delay is at most a given parameter $d$. 
Then it is possible to compute a minimum-delay schedule in time $O(\min(D^2(n),D(n)\log L))$.
\end{lemma}

\begin{proof}
We simulate the decision algorithm, as if it were given the (unknown) delay $d^*$ of a minimum-delay schedule. To do so, we maintain an open interval $(\ell,r)$ known to contain $d^*$; initially $(\ell,r)=(-\infty,\infty)$. Whenever the simulated decision algorithm performs a comparison of $d^*$ with some comparison value $c$, we simulate the comparison by checking whether $c$ belongs to the interval $(\ell,r)$, and (if it does) refining this interval to exclude $t$. Once $c$ lies outside $(\ell,r)$, we can determine the relative orders of $d^*$ and $c$ by comparing $c$ to $\ell$ and $r$.

To refine the interval $(\ell,r)$ to exclude $t$, we choose a test
value $t$ within the interval, and call both the decision algorithm
itself and the modified decision of Lemma~\ref{lem:modified-decision},
recursively with $t$ as their parameters. If the two algorithms
produce differing results, then $t$ is the optimal delay, and we
half the simulation and return $t$. If they determine that there
is a valid schedule with delay less than $t$, we set $r$ to $t$,
and the new interval to $(\ell,t)$.  And if they determine that
there is no schedule with delay $t$ or less, we set $\ell$ to $t$,
and the new interval to $(t,r)$.

It remains to specify how to choose the test value $t$ that we use to refine the interval. For each simulated comparison with a value $c$ within the interval, we perform at most two such tests. The first one selects $t$ to be the integer closest to the midpoint of the current interval $(\ell,r)$. If we refine the interval using that choice of $t$ and determine that $c$ still remains within the interval, then we perform a second refinement with $t=c$. The first choice of~$t$ ensures that the total number of refinement steps is $O(\log L)$, and the second choice of~$t$ ensures that, after these refinement steps, $c$ will be outside the remaining interval $(\ell,r)$.

The simulated algorithm behaves discontinuously at $d^*$ (it returns a valid schedule for larger values and a failure indication for smaller values). Because the only use it makes of its parameter is to perform comparisons, the only way it can be discontinuous at $d^*$ is to eventually perform a comparison in which the comparison value $c$ equals $d^*$. When it does so, the simulation will detect this equality and terminate the search with the optimal delay.
\end{proof}

\subsection{One-way crossings or Y merges}

The simplest example of our scheduling algorithm arises for two one-way roads that cross each other, with no platoons that turn from one road to the other. Each road has one incoming and one outgoing lane of traffic, the only pairs of incoming and outgoing lanes that are allowed to be used are the ones that stay on the same road, and traffic on one road cannot cross the intersection simultaneously with traffic on the opposite road.

Although it describes a different configuration of streets, this model is mathematically equivalent to one with two incoming lanes and one outgoing lane, forming a Y where the two incoming lanes merge. The pairs of lanes that are allowed are formed by one of the two incoming lanes together with the single outgoing lane. As before, it is not allowed for platoons from both incoming lanes to cross the merge point simultaneously. Almost the same mathematical model also applies to a T-junction of a minor two-way street onto a more major boulevard, restricted so that left turns from or to the boulevard are disallowed: right-turning traffic from the boulevard to the street, and through traffic on the far side of the boulevard from the street, can both flow freely, as they cannot interfere with any other platoons, and the remaining traffic has the same pattern as a Y merge.
(See Figure~\ref{fig:intersections}.)

\begin{figure}[hbt]
\centering
\begin{tabular}{ccc}
        \includegraphics[width=.3\textwidth]{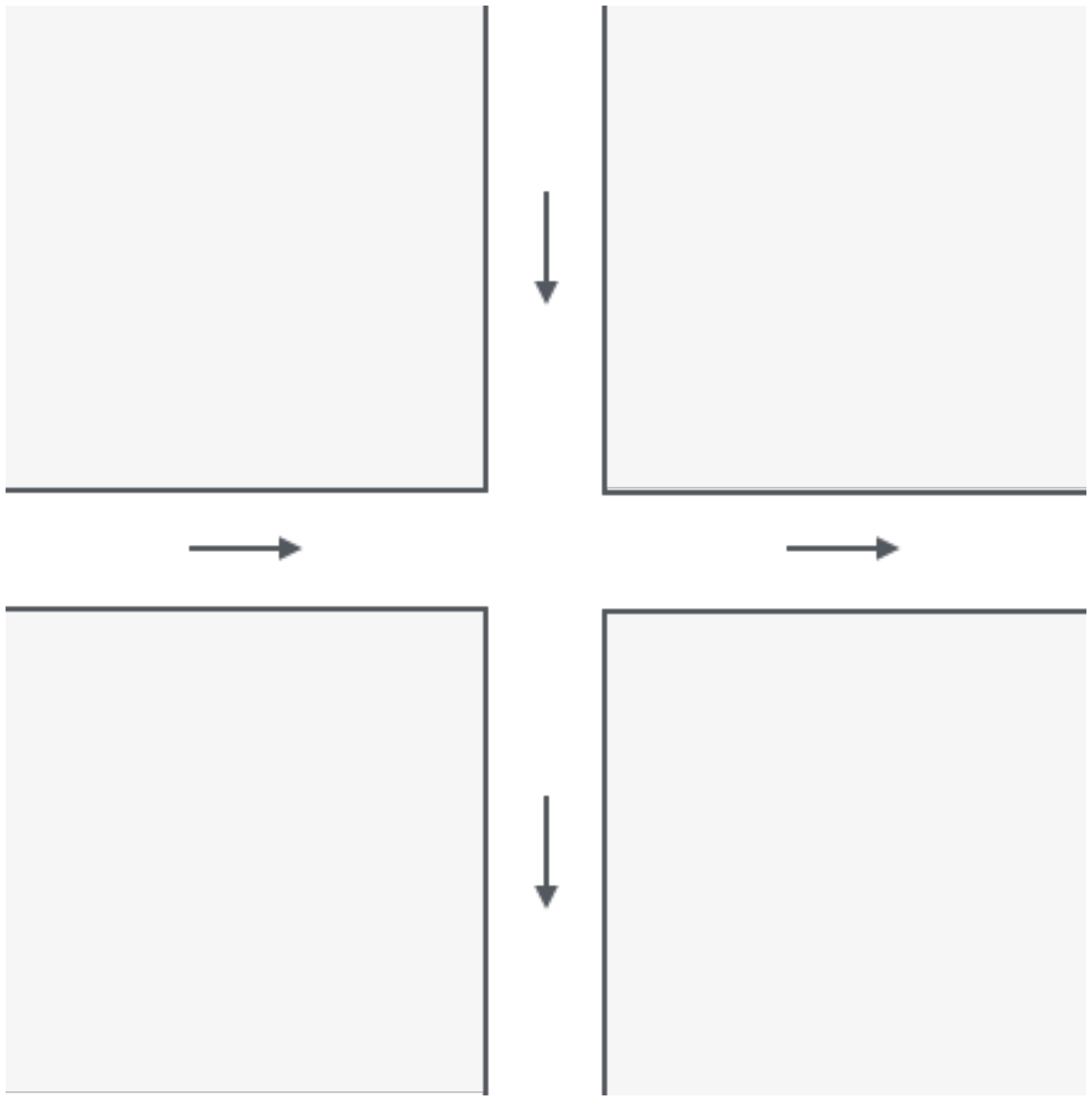}
&
        \includegraphics[width=.3\textwidth]{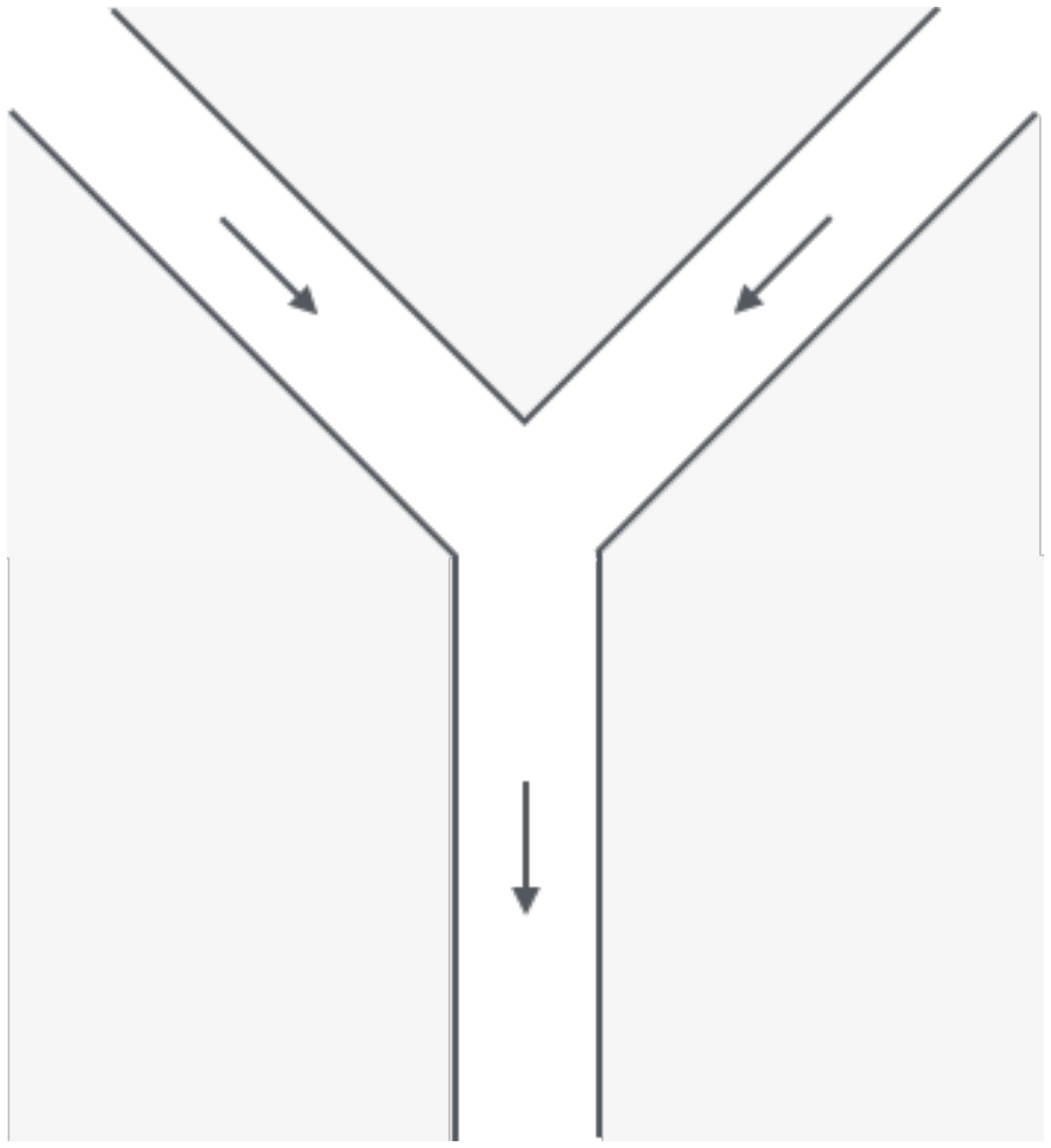}
&
        \includegraphics[width=.3\textwidth]{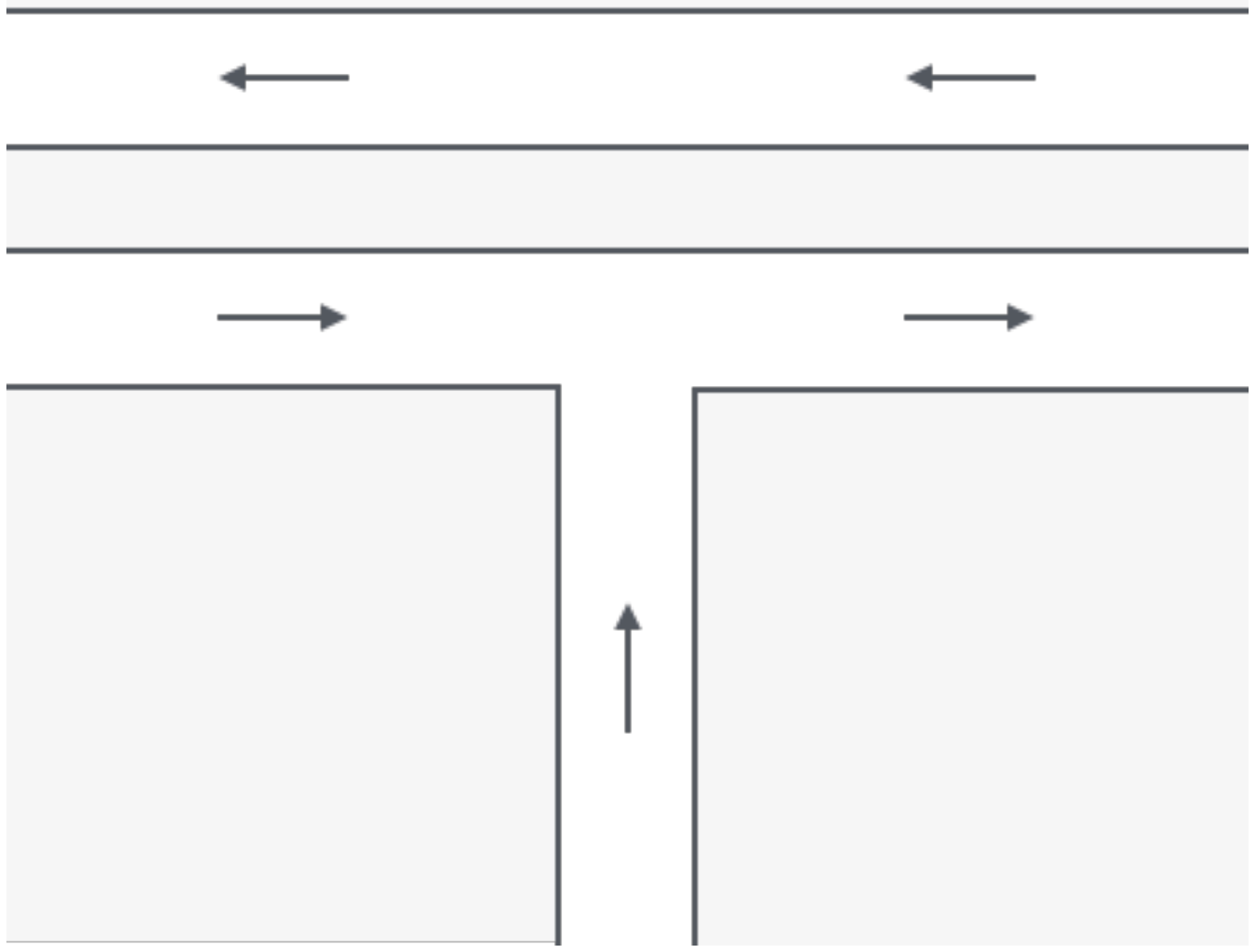}
\\[-.4in]
(a) & (b) & (c) 
\end{tabular}
\caption{Three intersections that are mathematically equivalent in our model:
(a) a crossing, (b) a Y merge, and (c) a T-junction.}
\label{fig:intersections}
\end{figure}

If we are trying to find a minimum-delay schedule, it is not always safe to allow a platoon to cross, even when there is no platoon on the other incoming lane that can cross at the same time. For example, consider the situation where one incoming road has a long platoon, ready to cross, while the other incoming road has a short platoon that is not yet ready but will reach the crossing soon. If we allow the long platoon to cross, the short platoon may be delayed for an excessive amount of time while waiting for the long platoon to finish crossing. On the other hand, if we delay the long platoon while we wait for the short platoon to arrive and cross, the delay for these two platoons may be better, but the time until both platoons have cleared the crossing will be longer, potentially causing greater delays for later platoons. 
(See Figure~\ref{fig:delay}.)
Nevertheless, if we know the maximum delay that we are willing to tolerate, we can apply a simple greedy algorithm that will either find a valid schedule with that delay or determine that no such schedule is possible.

\begin{figure}[hbt]
\centering
\label{fig:LongAndShort}
\includegraphics[width=0.4\textwidth]{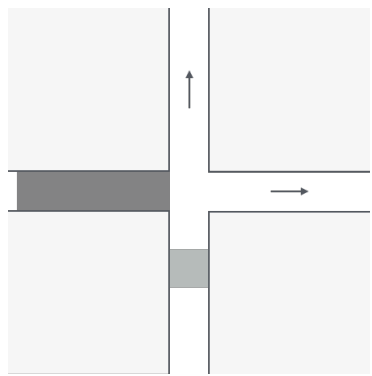}
\vspace*{-.4in}
\caption{An example where crossing if possible is not always the best choice. Delaying the crossing of the long platoon, until the short platoon has cleared the intersection, reduces the maximum delay of the schedule.}
\label{fig:delay}
\end{figure}

\begin{lemma}
\label{lem:greedy-Y}
Let $d$ be a delay parameter for a platoon scheduling problem with an intersection configured as a one-way crossing or Y merge and with $n$ platoons. Then an algorithm given $d$ as a parameter can either find a valid schedule with delay at most $d$, or determine that no such schedule is possible, in time $O(n)$.
\end{lemma}

\begin{proof}
After each platoon finishes crossing the intersection, the algorithm selects between the next two platoons, $p_i$ and $p_j$, to arrive at the intersection (one on each of the two incoming lanes), as follows. Let $p_i$ be the first of these platoons to arrive at the intersection (choosing arbitrarily when they both arrive at the same time). If allowing $p_i$ to cross as soon as it can would delay $p_j$ by at most $d$ time units, then $p_i$ is allowed to cross next. Otherwise, $p_j$ crosses next. If the resulting schedule has delay at most $d$, it is returned; otherwise, the algorithm reports that no such schedule is possible.

Clearly, this simple algorithm takes time $O(n)$ and, when it returns a valid schedule, the schedule has delay at most $d$. It remains to show that, if a valid schedule $S$ with delay at most $d$ exists, then our algorithm will succeed in finding a schedule (possibly different from $S$) that also has delay at most $d$. We may assume without loss of generality that (like our greedy algorithm) $S$ schedules each platoon as early as possible given the ordering of platoons across the crossing that it selects.
We may also assume without loss of generality that the hypothetical valid schedule $S$ follows the same sequence of scheduling choices as our greedy algorithm for as many steps as possible. We will prove by contradiction that, with this assumption, $S$ must actually equal our greedy schedule.

For, if not, $S$ and the greedy schedule diverge at some point $t$ in time, when two platoons $p_i$ and $p_j$ are arriving on the two incoming lanes, $S$ chooses one of them as the next to cross, and our greedy algorithm selects the other one as the next to cross. Let $p_i$ be the first of these platoons that our greedy algorithm considers as a candidate for the next one to cross. Then our algorithm will only choose $p_j$ if it is forced to (because choosing $p_i$ would cause an excessive delay to $p_j$), and in this case $S$ cannot choose $p_i$ for the same reason. So the only way for our algorithm and $S$ to differ would be for our algorithm to allow $p_i$ to cross and for $S$ to instead let $p_j$ be the next platoon to cross. But in this case let $S'$ be a schedule modified from $S$ by allowing $p_i$ to cross next, and otherwise keeping all platoons in the same order given by $S$. The platoons that are disadvantaged by this change are $p_j$ and the other platoons on the same incoming lane that immediately follow $p_j$ in schedule $S$, but their maximum delay in $S'$ is at most $d$. For all remaining platoons, this change in schedule does not cause any additional delays, because the total time until $p_i$, $p_j$, and the other platoons following $p_j$ in the same lane have all crossed can only decrease because of the earlier release time of $p_i$ relative to $p_j$. So, like $S$, schedule $S'$ also has maximum delay at most $d$, but it agrees with our greedy algorithm for one more step. This contradicts the choice of $S$ as the schedule that agrees with the greedy algorithm for as many steps as possible, and the contradiction can only be resolved by $S$ (a valid schedule with delay at most $d$) equalling the greedy schedule.
\end{proof}

\begin{theorem}
A minimum-delay schedule for a platoon scheduling problem with an intersection configured as a one-way crossing or Y merge and with $n$ platoons can be found in time $O(\min(n^2,n\log L))$.
\end{theorem}

\begin{proof}
We apply the parametric search technique of Lemma~\ref{lem:parametric-search},
using the greedy algorithm of Lemma~\ref{lem:greedy-Y} 
as the decision algorithm.
\end{proof}

\subsection{Multiway merges}

It is not clear how to extend our greedy scheduling decision algorithm even to $3$-way merges. For instance, consider a situation where a long platoon arrives on one incoming lane, somewhat earlier than two shorter platoons would arrive on two other lanes. Even if the long platoon would not necessarily cause excessive delays to either short platoon by itself, allowing the long platoon to cross might lead to a situation where neither of the two short platoons can cross, because it would excessively delay the other one. Untangling these indirect effects seems beyond the scope of the local decisions made by the greedy algorithm.

Nevertheless, we can find an optimal schedule for a $k$-way merge in polynomial time using a somewhat more complicated dynamic programming algorithm. We model the intersection as having $k$ incoming lanes and one outgoing lane. The outgoing lane can be paired with any incoming lane, but only one such pair of an incoming and outgoing lane can use the intersection at any given time.
(See Figure~\ref{fig:k-way}.)

\begin{figure}[hbt]
\centering

\label{fig:5WayIntersection}
\vspace*{-.1in}
\includegraphics[width=0.6\textwidth, trim= 0in 2.3in 0in 2.3in, clip]{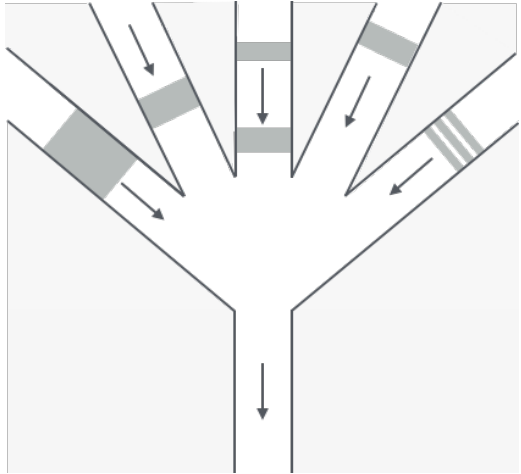}
\vspace*{-.1in}
\caption{An example 5-way merge intersection.}
\label{fig:k-way}
\end{figure}

We define a \emph{state} of an intersection to be a situation in which some platoons have completely crossed the intersection and some others still remain to cross, without there being a platoon that is only partly across. For a scheduling task with $n$ platoons and $k$ incoming lanes, there are $O(n^k)$ possible states that could occur. Any actual schedule for this task can be represented as a sequence of $n+1$ states, starting from a state in which no platoons have crossed and ending at a state in which all platoons have crossed. For any such sequence of states, it is safe to allow each platoon to cross as early as possible, consistent with the ordering of the platoons determined by the sequence of states.

\begin{lemma}
\label{lem:kway-decision}
For a $k$-way merge platoon scheduling problem as described above (with $k$ constant), and a given parameter $d$, it is possible to determine in time $O(n^k)$ whether a valid schedule with delay at most $d$ exists.
\end{lemma}

\begin{proof}
We use dynamic programming to find, for each state $s$, the earliest time $t_d(s)$ that it is possible to reach state $s$ via a partial schedule in which the maximum delay of any platoon that crosses the intersection within the partial schedule is $d$ (or $+\infty$ if no such schedule exists). As a base case, for the state in which no platoons have crossed, $t_d(s)$ may be set to the earliest release time of any platoon.

For each state $s$, there are (at most) $k$ states $s_1,\dots s_k$ that could be the predecessor of $s$ in a valid sequence of states, obtained from $s$ by omitting the last platoon to cross on each of the $k$ incoming lanes (if $s$ includes a platoon that has already crossed on that lane). A potential schedule for $s$ may be obtained by choosing an incoming lane $i$, choosing a schedule for $s_i$ obtaining the earliest possible completion time $t_d(s_i)$, and then allowing the final platoon on lane $i$ to cross at the maximum of $t_d(s_i)$ and its release time. If this potential schedule does not delay the platoon on lane~$i$ by more than $d$, it is valid. The earliest completion time $t_d(s)$ may be computed by finding all valid schedules of this type and choosing one for which the final platoon finishes crossing as early as possible.

The overall algorithm loops through the states in a consistent ordering, chosen so that for each state $s$ the predecessor states $s_i$ will all already have been looped through. For each state $s$ in this loop, it uses the computation described above to compute $t_d(s)$. The time is constant for each state, and there are $O(n^k)$ states, so the total time is $O(n^k)$.

A valid schedule for the whole scheduling task exists if and only if the state $s$ representing the situation in which all platoons have crossed has a finite value of $t_d(s)$.
\end{proof}

\begin{theorem}
A minimum-delay schedule for a platoon scheduling problem with an intersection configured as a $k$-way merge and with $n$ platoons can be found in time $O(\min(n^{2k},n^k\log L))$.
\end{theorem}

\begin{proof}
We apply the parametric search technique of Lemma~\ref{lem:parametric-search},
using the dynamic programming algorithm of 
Lemma~\ref{lem:kway-decision} as the decision algorithm.
\end{proof}

\subsection{Two-way crossing}

Our most complicated single-intersection model has two roads with two-way traffic, crossing each other at a single intersection, with no left turns allowed. There are four incoming lanes of traffic paired with four outgoing lanes of traffic. Two pairs of lanes on the same road as each other do not interfere (platoons of traffic traveling on these pairs of lanes can simultaneously pass through the intersection without delays) but any traffic on one road interferes with all traffic on the other road.
(See Figure~\ref{fig:two-way}.)

\begin{figure}[hbt]
\centering
\includegraphics[width=0.45\textwidth, trim= 0in 2in 0in 2in, clip]{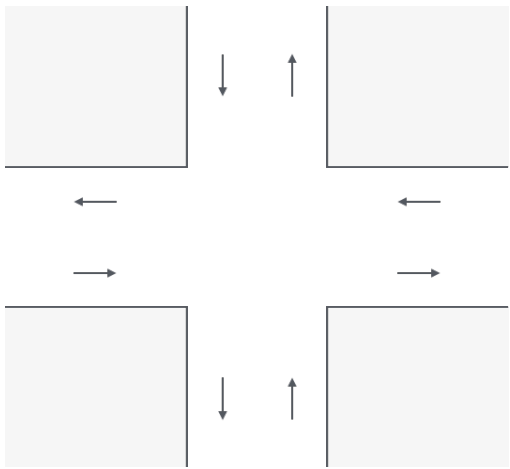}
\caption{A two-way crossing.}
\label{fig:two-way}
\end{figure}

As for the $k$-way merge, we define a state to be a situation in which some platoons have completely crossed the intersection and some others still remain to cross, without there being a platoon that is only partly across. There are $O(n^4)$ states, one for each way of selecting an initial subset of the platoons on each of the four incoming lanes. Unlike the $k$-way merge, however, a valid schedule for all the platoons does not necessarily have $n+1$ states, differing from each other by a single platoon. Instead, the states can be guaranteed to occur within a valid schedule only at times when the schedule switches from allowing traffic to cross the intersection on one of the roads to allowing traffic to cross on the other road. Because the parts of the schedule between two states are more complicated, the dynamic programming algorithm for stringing together states into an optimal schedule is also more complicated.

To be specific, we compute (as before) the minimum time $t_d(s)$ at which a valid schedule can reach state $s$, with maximum delay at most $d$ on the platoons that cross the intersection as part of state~$s$.  To compute $t_d(s)$, we examine each state $s'$ that differs from state $s$ only by traffic on (all four lanes of) a single road. A valid schedule for $s$ can be obtained from the optimal schedule for $s'$ (the schedule that achieves completion time $t_d(s')$ by greedily scheduling the remaining traffic by which $s$ differs from $s'$, scheduling each platoon as soon as it is released or otherwise available to cross the intersection, as long as this greedy schedule also achieves maximum delay at most $d$. The optimal completion time $t_d(s)$ is the minimum, over all of the $O(n^2)$ potential predecessor states $s'$, of the completion time obtained by appending this greedy schedule (whenever it is valid) to the optimal schedule for~$s'$.

\begin{lemma}
Given a state $s$ and a potential predecessor state $s'$, form a schedule for $s$ by appending a greedy schedule for the remaining cars to a schedule for $s'$ that obtains completion time $t_d(s')$. Then it is possible to test whether the schedule for $s$ obtained in this way has maximum delay at most $d$, and to compute the completion time of the resulting schedule, in time $O(1)$.
\end{lemma}

\begin{proof}
By assumption, the part of the schedule for the platoons in $s'$ has maximum delay at most $d$. Among the remaining platoons, the most heavily delayed will be the first ones to go in the two lanes controlled by the greedy schedule. For each of these two platoons, we can calculate its delay as $\max(0,t_d(s')-r_i)$ where $r_i$ is the release time of the platoon.

The completion time of the schedule is the maximum, over the two lanes controlled by the greedy part of the schedule, of the end time of the last platoon plus the amount by which that platoon was delayed. If the delay of the first greedily-scheduled platoon on the lane is $d$, then the delay of the last platoon on the same lane is $\max(0,d-\sum g_i)$ where the numbers $g_i$ are the gaps between the end time of one platoon and the start time of the next platoon, for the platoons scheduled on that lane by the greedy algorithm. If we store the prefix sums of the gaps, we can calculate the sum of the gaps for any contiguous interval of platoons in constant time, by subtracting the prefix sum up to the first platoon from the prefix sum for the last platoon.
\end{proof}

\begin{lemma}
\label{lem:2wc-decision}
For a two-way crossing platoon scheduling problem as described above, and a given parameter $d$, it is possible to determine in time $O(n^6)$ whether a valid schedule with delay at most $d$ exists.
\end{lemma}

\begin{proof}
There are $O(n^4)$ states~$s$, $O(n^2)$ predecessor states $s'$ per state, and $O(1)$ time to perform the greedy scheduling algorithm that augments a schedule for $s'$ to a schedule for~$s$.
Multiplying these terms together gives $O(n^6)$.
\end{proof}

\begin{theorem}
A minimum-delay schedule for a platoon scheduling problem with an intersection configured as a two-way crossing and with $n$ platoons can be found in time $O(\min(n^{12},n^6\log L))$.
\end{theorem}

\begin{proof}
We apply the parametric search technique of 
Lemma~\ref{lem:parametric-search},
using the dynamic programming algorithm of 
Lemma~\ref{lem:2wc-decision} as the decision algorithm.
\end{proof}


\section{Hardness}
In this section, we show that combining the two way intersection
and multiway merge versions leads to an NP-complete version of the
problem.  Specifically, let us consider a version 
of the problem where an arbitrary number of
lanes are merging onto one outgoing lane, one lane of traffic is
going in the opposite direction, and one lane of traffic crosses
these two.  The multilane merge and the lane of traffic in the
opposite direction can both use the intersection simultaneously,
but neither of them can use the intersection when a platoon on the
third lane travels through the intersection.
(See Figure~\ref{fig:reduction}.)

\begin{figure}[hbt]
\centering
\includegraphics[width=0.5\textwidth, trim=0in 1in 0in 1in, clip]{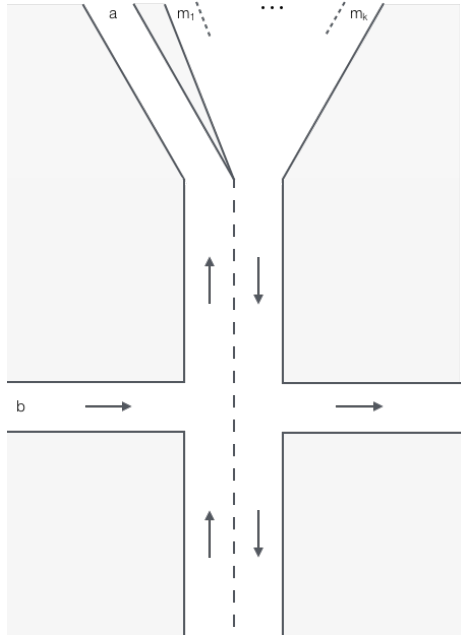}
\caption{The intersection used in our hardness proof.}
\label{fig:reduction}
\end{figure}

To precisely specify the MULTI-CROSS problem, we name the roads as follows:
\begin{itemize}
\item The multiway merge has $k$ incoming lanes $m_1,m_2,\dots,m_k$ for some input parameter $k$
\item the street parallel to the multiway merge is $a$
\item the street crossing these two is $b$.
\end{itemize}
No two platoons one on $m_i$ and one on $m_j$ can enter the merge at the same time.  If a platoon on $b$ is in the intersection, no other platoons can be passing through as well.  The platoons on $a$ do not interfere with any platoon on an $m_i$ lane.

An instance of the MULTI-CROSS problem is defined by a value for $k$, a set of platoons $P = [ (r_1,s_1,t_1),\dots (r_n,s_n,t_n)]$ assigned to the roads and their arrival and exit times at the intersection (platoon $p_i$ is on road $r_i$ and arrives at the intersection at time $s_i$ and would exit the intersection at time $t_i$ if there is no delay), and a maximum delay parameter $d_{\max}$. The decision problem is to decide whether or not there is a schedule of the platoons through the intersection such that no platoon experiences a delay more than $d_{\max}$.  Any schedule can be simulated to check for validity and compute the maximum delay; hence, MULTI-CROSS is in NP.

To show the MULTI-CROSS problem is NP-hard, we reduce the \textsc{PARTITION} problem to it.  The \textsc{PARTITION} problem is to given a multiset of positive integers $X = \lbrace x_1,\dots,x_\ell\rbrace$, decide whether or not they can be partitioned into two sets $U$ and $V$ such that $\sum_{x\in U} x = \sum_{x\in V}x$.  \textsc{PARTITION} is known to be NP-complete~\cite{Garey:1990:CIG:574848}.

Given an instance of the \textsc{PARTITION} problem $X$, we build an instance of the MULTI-CROSS problem as follows.  Let $q=\frac{\sum_{x \in X} x}{2}$ and set $d_{\max} = 2q+1$ and $k = \ell+1$.  We place one long platoon on $a$ and one short platoon on $b$: $p_1 = (a,0,4(q+1))$ and $p_2 = (b,2q,2q+1)$.  Then one other long platoon is placed on one of the multiway merge lanes: $p_3 = (m_{\ell+1},q,q+4(q+1))$.  Finally for each integer $x_i \in X$, place one platoon on lane $m_i$: $p_{3+i} = (m_i,q,q+x_i)$.

\begin{lemma}\label{lem:part-to-sched}
If there is a valid partitioning of $X$, then there is a viable schedule with maximum delay $2q+1$.
\end{lemma}
\begin{proof}
Let $U$ and $V$ be such a partitioning of $X$. Assign the platoons who correspond to integers in $U$ go through the merge first.  Once the short platoon $p_2$ arrives on road $b$ have it cross immediately.  At this point platoon $p_1$ will have been waiting for exactly $q+1$ units of time and must cross immediately after $p_2$.  Simultaneously set the platoons associated with the integers in $V$ to merge.  Finally once every other platoon has merged, send platoon $p_3$ through the merge.

Platoon $p_1$ has a delay of exactly $2q+1$ in this schedule, $p_2$ has a delay of $0$, and platoon $p_3$ also has a delay of exactly $2q+1$.  The platoons $p_4, \dots , p_{3+\ell}$ all had delays less than $2q+1$.  Therefore the maximum delay of this schedule is $2q+1$.
\end{proof}

\begin{lemma}\label{lem:sched-to-part}
If there is a schedule with maximum delay at most $2q+1$, then there is a valid partitioning of $X$.
\end{lemma}
\begin{proof}
If any platoon blocks an intersection for more than $2q+1$ units of time while another platoon is waiting, then the maximum delay must be more than $2q+1$.  Therefore $p_1$ cannot go until $p_2$ passes and $p_3$ cannot go until $p_4,\dots,p_{3+\ell}$ all go through the merge.  Because $p_1$ has been waiting for $2q$ time units when $p_2$ arrives at the intersection, $p_2$ must immediately enter the intersection otherwise $p_1$ will delay more than $2q+1$ time units.  The platoons $p_4,\dots,p_{3+\ell}$ must all clear the intersection before time $3q+1$ for $p_3$ to enter the intersection with a delay of at most $2q+1$.  They arrive at time $q$, have $q$ time before $p_2$ must go through, and have $q$ time after before $p_3$ must go through.  Therefore if it is possible to route these platoons through their merge, then the two sets of platoons who go through the merge before or after $p_2$ each must sum to exactly $q$.  So the two sets of integers these two sets of platoons correspond to are a valid partitioning of $X$.
\end{proof}

\begin{theorem}
The MULTI-CROSS problem is NP-complete.
\end{theorem}
\begin{proof}
By Lemmas~\ref{lem:part-to-sched}~and~\ref{lem:sched-to-part}, and
our observation that MULTI-CROSS is in NP.
\end{proof}

\section{Conclusion and Future Work}
We have studied several scheduling problems for routing
autonomous vehicle platoons through an unregulated intersection, providing
polynomial-time algorithms for the cases of a $k$-way merge (for constant $k$)
and for a crossing involving two-way traffic.
We have also provided an NP-completeness result for instances of the problem
that involve a $k$-way merge (for non-constant $k$) and two-way traffic.
We leave as open problem to determine whether the $k$-way merge version
of problem for non-constant $k$ is NP-complete or whether there is a
polynomial-time algorithm for solving this problem.
In addition, we studied offline versions of all these problems, and it may
be interesting to study online versions, where the platoons and 
their desired paths are not all known in advance.

{\raggedright
\bibliographystyle{plainurl}
\bibliography{refs}

\begin{thebibliography}{10}

\bibitem{A6859163}
H.~Ahn, A.~Colombo, and D.~Del Vecchio.
\newblock Supervisory control for intersection collision avoidance in the
  presence of uncontrolled vehicles.
\newblock In {\em 2014 American Control Conf.}, pages 867--873, 2014.

\bibitem{altche2016time}
F.~Altch{\'e}, X.~Qian, and A.~de~La~Fortelle.
\newblock Time-optimal coordination of mobile robots along specified paths.
\newblock {\em arXiv preprint arXiv:1603.04610}, 2016.

\bibitem{A4020351}
G.~Antonelli and S.~Chiaverini.
\newblock Kinematic control of platoons of autonomous vehicles.
\newblock {\em IEEE Trans. on Robotics}, 22(6):1285--1292, 2006.

\bibitem{au2010motion}
T.~Au and P.~Stone.
\newblock Motion planning algorithms for autonomous intersection management.
\newblock In {\em AAAI Workshop on Bridging the Gap Between Task and Motion
  Planning}, 2010.

\bibitem{B1411418}
J.~Baber, J.~Kolodko, T.~Noel, M.~Parent, and L.~Vlacic.
\newblock Cooperative autonomous driving: intelligent vehicles sharing city
  roads.
\newblock {\em IEEE Robotics Automation Magazine}, 12(1):44--49, 2005.

\bibitem{Berger:2010}
F.~Berger and R.~Klein.
\newblock A traveller's problem.
\newblock In {\em 26th ACM Symp. on Computational Geometry (SoCG)}, pages
  176--182, 2010.

\bibitem{C6728285}
D.~Carlino, S.~D. Boyles, and P.~Stone.
\newblock Auction-based autonomous intersection management.
\newblock In {\em 16th Int. IEEE Conf. on Intelligent Transportation Systems
  (ITSC)}, pages 529--534, 2013.

\bibitem{Dasler2015}
P.~Dasler and D.~Mount.
\newblock On the complexity of an unregulated traffic crossing.
\newblock In Frank Dehne, J{\"o}rg-R{\"u}diger Sack, and Ulrike Stege, editors,
  {\em 14th Int. Symp. on Alg. and Data Struct. (WADS)}, pages 224--235, 2015.
\newblock see also \url{http://arxiv.org/abs/1505.00874}.

\bibitem{Dresner:2004}
K.~Dresner and P.~Stone.
\newblock Multiagent traffic management: A reservation-based intersection
  control mechanism.
\newblock In {\em 3rd Int. Conf. on Autonomous Agents and Multiagent Systems
  (AAMAS)}, pages 530--537, 2004.

\bibitem{Dresner:2005}
K.~Dresner and P.~Stone.
\newblock Multiagent traffic management: An improved intersection control
  mechanism.
\newblock In {\em 4th Int. Conf. on Autonomous Agents and Multiagent Systems
  (AAMAS)}, pages 471--477, 2005.

\bibitem{dresner2008multiagent}
K.~Dresner and P.~Stone.
\newblock A multiagent approach to autonomous intersection management.
\newblock {\em Journal of Artificial Intelligence Research}, pages 591--656,
  2008.

\bibitem{F1429220}
E.~Frazzoli and F.~Bullo.
\newblock Decentralized algorithms for vehicle routing in a stochastic
  time-varying environment.
\newblock In {\em 43rd IEEE Conf. on Decision and Control (CDC)}, volume~4,
  pages 3357--3363 Vol.4, 2004.

\bibitem{Garey:1990:CIG:574848}
M.~R. Garey and D.~S. Johnson.
\newblock {\em Computers and Intractability; A Guide to the Theory of
  {NP}-Completeness}.
\newblock W. H. Freeman \& Co., New York, NY, USA, 1990.

\bibitem{IlginGuler2014121}
S.~Ilgin Guler, M.~Menendez, and L.~Meier.
\newblock Using connected vehicle technology to improve the efficiency of
  intersections.
\newblock {\em Transportation Research Part C: Emerging Technologies},
  46:121--131, 2014.

\bibitem{hatzack2014operational}
W.~Hatzack and B.~Nebel.
\newblock The operational traffic control problem: Computational complexity and
  solutions.
\newblock In {\em Sixth European Conference on Planning}, pages 113--119, 2014.

\bibitem{HEARN200572}
R.~A. Hearn and E.~D. Demaine.
\newblock {PSPACE}-completeness of sliding-block puzzles and other problems
  through the nondeterministic constraint logic model of computation.
\newblock {\em Theoretical Computer Science}, 343(1):72--96, 2005.

\bibitem{L6121907}
J.~Lee and B.~Park.
\newblock Development and evaluation of a cooperative vehicle intersection
  control algorithm under the connected vehicles environment.
\newblock {\em IEEE Transactions on Intelligent Transportation Systems},
  13(1):81--90, 2012.

\bibitem{L5940562}
J.~Levinson, J.~Askeland, J.~Becker, J.~Dolson, D.~Held, S.~Kammel, J.~Z.
  Kolter, D.~Langer, O.~Pink, V.~Pratt, M.~Sokolsky, G.~Stanek, D.~Stavens,
  A.~Teichman, M.~Werling, and S.~Thrun.
\newblock Towards fully autonomous driving: Systems and algorithms.
\newblock In {\em IEEE Intelligent Vehicles Symp. (IV)}, pages 163--168, 2011.

\bibitem{L6957676}
G.~Lu, L.~Li, Y.~Wang, R.~Zhang, Z.~Bao, and H.~Chen.
\newblock A rule based control algorithm of connected vehicles in uncontrolled
  intersection.
\newblock In {\em 17th Int. IEEE Conf. on Intelligent Transportation Systems
  (ITSC)}, pages 115--120, 2014.

\bibitem{Meg-JACM-83}
N.~Megiddo.
\newblock {Applying parallel computation algorithms in the design of serial
  algorithms}.
\newblock {\em J. ACM}, 30(4):852{--}865, 1983.

\bibitem{M7039600}
D.~Miculescu and S.~Karaman.
\newblock Polling-systems-based control of high-performance provably-safe
  autonomous intersections.
\newblock In {\em 53rd IEEE Conf. on Decision and Control}, pages 1417--1423,
  2014.

\bibitem{N683216}
R.~Naumann, R.~Rasche, and J.~Tacken.
\newblock Managing autonomous vehicles at intersections.
\newblock {\em IEEE Intelligent Systems and their Applications}, 13(3):82--86,
  1998.

\bibitem{Rajamani200115}
R.~Rajamani and S.~E. Shladover.
\newblock An experimental comparative study of autonomous and co-operative
  vehicle-follower control systems.
\newblock {\em Transportation Research Part C: Emerging Technologies},
  9(1):15--31, 2001.

\bibitem{S1468255}
D.~J. Stilwell, B.~E. Bishop, and C.~A. Sylvester.
\newblock Redundant manipulator techniques for partially decentralized path
  planning and control of a platoon of autonomous vehicles.
\newblock {\em IEEE Trans. on Systems, Man, and Cybernetics, Part B
  (Cybernetics)}, 35(4):842--848, 2005.

\bibitem{VanMiddlesworth:2008}
M.~VanMiddlesworth, K.~Dresner, and P.~Stone.
\newblock Replacing the stop sign: Unmanaged intersection control for
  autonomous vehicles.
\newblock In {\em 7th Int. Conf. on Autonomous Agents and Multiagent Systems
  (AAMAS)}, pages 1413--1416, 2008.

\bibitem{Wuthishuwong2015}
C.~Wuthishuwong, A.~Traechtler, and T.~Bruns.
\newblock Safe trajectory planning for autonomous intersection management by
  using vehicle to infrastructure communication.
\newblock {\em EURASIP Journal on Wireless Communications and Networking},
  2015(1):1--12, 2015.

\end{thebibliography}
}

\end{document}